\documentclass[journal]{IEEEtran}
\usepackage{amsmath}
\usepackage{multicol}
\usepackage{mathtools}
\usepackage{amssymb}
\usepackage{amsmath}  
\usepackage{subeqnarray}
\usepackage{bm}
\usepackage{color}

\usepackage{algorithm}
\usepackage{algorithmic}        

\usepackage{cite}

\usepackage{array}  
\usepackage{url}

\usepackage{amsthm}
\usepackage{graphicx}
\usepackage{subfigure}
\theoremstyle{plain}     
\newtheorem{thm}{Theorem}

\newtheorem{cor}{Corollary}
\theoremstyle{definition}

\theoremstyle{remark}

\begin{document}
%
\title{Effective Capacity Analysis in Ultra-Dense Wireless Networks with Random Interference}

\author{\IEEEauthorblockN{Yu Gu~\IEEEmembership{Student Member,~IEEE,} Qimei Cui,~\IEEEmembership{Senior Member,~IEEE,}
         Yu Chen,~\IEEEmembership{Member,~IEEE,}\\
         Wei Ni,~\IEEEmembership{Senior Member,~IEEE,}
         Xiaofeng Tao,~\IEEEmembership{Senior Member,~IEEE,}
         Ping Zhang,~\IEEEmembership{Senior Member,~IEEE}}
\thanks{Y. Gu, Q. Cui, Y. Chen, X. Tao, and P. Zhang are with the National Engineering Laboratory for Mobile Network Technologies, Beijing University of Posts and Telecommunications, Beijing 100876, China (Corresponding author:
Qimei Cui; cuiqimei@bupt.edu.cn).}
\thanks{W. Ni is with the Digital Productivity and Services Flagship, Commonwealth Scientific and Industrial Research Organization (CSIRO), Sydney, N.S.W. 2122, Australia (e-mail: wei.ni@csiro.au).}}

\maketitle

\begin{abstract}
Ultra-dense networks (UDNs) provide a promising paradigm to cope with exponentially increasing mobile traffic. However, little work has to date considered unsaturated traffic with quality-of-service (QoS) requirements. This paper presents a new cross-layer analytical model to capture the unsaturated traffic of a UDN in the presence of QoS requirements. The effective capacity (EC) of the UDN is derived, taking into account small-scale channel fading and possible interference. Key properties of the EC are revealed. The amount of traffic impacts effective capacity of the UDN due to the sophisticated interactions among small base stations operating in the same frequency. The maximization of total effective capacity is formulated as a non-cooperative game in the paper. The best-response function is derived , iteratively searching the Nash equilibrium point. System simulation results indicate that our proposed model is accurate. The simulations also show the maximum allowed arrival rate with the QoS guarantee, compared with the full interference model.
\end{abstract}


%
\IEEEpeerreviewmaketitle

\section{Introduction}
\IEEEPARstart{U}{ltra}-dense network (UDN) is a promising technology to address the ever and exponentially-increasing demand for mobile traffic \cite{UDNsurvey,cui2017preserving}. In UDNs, small base stations (SBSs) are densely deployed to allow user equipments (UEs) to be in proximity of SBSs. This can help improve signal qualities of the UEs. Moreover, the network densification can improve network capacity by allowing frequency to be spatially reused. On the other hand, the ultra-dense deployment would cause a severe issue of interference \cite{Interference}. The interference distribution in UDNs is more unpredictable and unmanageable. Accurate characterization of the interference in UDNs is a fundamental step towards understanding the overall performance of dense networks.

Studies have been carried out to tackle the interference of ultra-dense networks, typically under the assumption of a spatial Poisson point process (PPP) in the placement of SBSs with saturated traffic. Analyses have been conducted on coverage probability, throughput and other performance metrics in the downlink of UDNs \cite{PPP1,PPP2}. However, these studies focus on the scenario where SBSs always have packets to transmit (i.e., saturated traffic). However, in practice, unsaturated traffic can have a strong impact on the distribution of interference in UDNs. This is because SBSs without traffic data can remain idle to prevent interference to other SBSs. Yu \cite{INACTIVE1} and Kibria \cite{INACTIVE2} considered an inactive BS probability using the stochastic geometry, which defined the probability that a randomly chosen BS did not have any UE in its Voronoi cell. Arvanitakis \cite{GLoFlow} developed an analytical model to capture load-based interference, assuming that the load of a SBS was $\rho$ ($\rho <1$) and the SBS caused interference only $\rho$ of the time. However, the interference is time-varying according to the time-varying traffic. Little works can capture the random interference due to the unsaturated traffic. The performance studies based on PPP \cite{PPP1,PPP2,INACTIVE1,INACTIVE2,GLoFlow} could be restrictive. The consideration of traffic randomness is important to accurately analyze the performance of UDNs.

Another challenge in UDNs is how to guarantee the quality-of-service (QoS) of UEs, or more specifically, delay. Particularly, the time-varying wireless channel quality and random traffic arrival can significantly affect the signal-to-interference-plus-noise-ratio (SINR) at the UEs. To analyze the time-varying, random interference with unsaturated traffic and the QoS, there are a small number of studies on networks performance with aid of the queueing model \cite{Queue2,Queue1}. \cite{Queue2} developed a 3-dimensional (3-D) Markov chain models to capture the queue length and important QoS measures, such as delay, packet loss, and throughput, for practical 802.11 systems with finite buffer under finite load. \cite{Queue1} proposed an N-dimensional Markov chain to model the queue states. The queue states decided the interfering SBSs sets. The stationary probability of the queue length of each UE could be obtained, which enabled the calculations of the average delay and the average packet overflow probability. Most of the above works characterize the average queueing performance (i.e., average delay) in a deterministic way, and suppress the variations of wireless mobile fading channels. Meanwhile, satisfying a deterministic delay constraint is either different to achieve or at a cost of substantially reducing transmission rates \cite{7061966}.

In contrast to the above deterministic delay QoS bounds, effective capacity (EC) is a powerful tool for characterizing physical-layer channels under statistical QoS guarantee \cite{EC}. The analysis and application of EC in various scenarios have attracted lots of interest, such as cellular network \cite{ChenEC} and random access network \cite{CuiEC}. However, most existing works have focused on time-varying fading channels, and overlooked the time-varying interference caused by unsaturated traffic at the SBSs. In general, non-trivial effort would be required to have the EC extended to UDN scenarios.

In this paper, we investigate the prominent region between capacity and QoS in UDNs, under unsaturated traffic and statistically characterized QoS requirements. Unsaturated traffic can have a strong impact on the distribution of interference in UDNs. Hence, we develop a novel cross-layer analytical model for UDNs to capture the small-scale channel fading and possible interference. Comparing to the recent studies, breakthroughs are brought forth in developing the theory of the effective capacity in UDNs. The main contributions of this paper can be summarized as follows:
\begin{itemize}
  \item First, we develop a cross-layer analytical model to calculate the idle probability of SBSs and the EC of UDNs. For the two--SBSs case, we derive the probability density function (PDF) of SINR and, in turn, the EC of the SBSs in the presence of unsaturated traffic and non-trivial QoS requirements.
  \item Second, we derive the EC for N--SBSs case, considering the $2^{(N-1)}$ sets account for all possible interfering SBSs and the multiple integrals accounting for small-scale fading. And some key properties of the EC are provided.
  \item Thirdly, to characterize the impact of the level of traffic saturation on the total effective capacity of UDN, a non-cooperative game is formulated to maximize QoS- guaranteed total average arrival traffic for UDN. The best-response function is derived and can be used to iteratively search the Nash equilibrium (NE) point.
  \item Last but not least, our simulations show that our proposed model is accurate, compared with existing models. Numerical simulations show the maximum traffic arrival rate with the QoS guarantee in UDNs, and an appreciable improvement of traffic arrival rate is observed as compared to the saturated traffic condition. Simulations show using the best-response function can coverage to the NE point at finite iterations, which is superior to other strategies of traffic unsaturation.
\end{itemize}

The remainder of this paper is organized as follows. Section II describes the system model. In Section III, we derive the new EC under unsaturated traffic, and infer key properties. In Section IV, we develop a simulation platform and validate our new analytical model, followed by conclusions in Section V.
\section{System Model}
\begin{figure}[!t]
\centering
\includegraphics[width=3.5in]{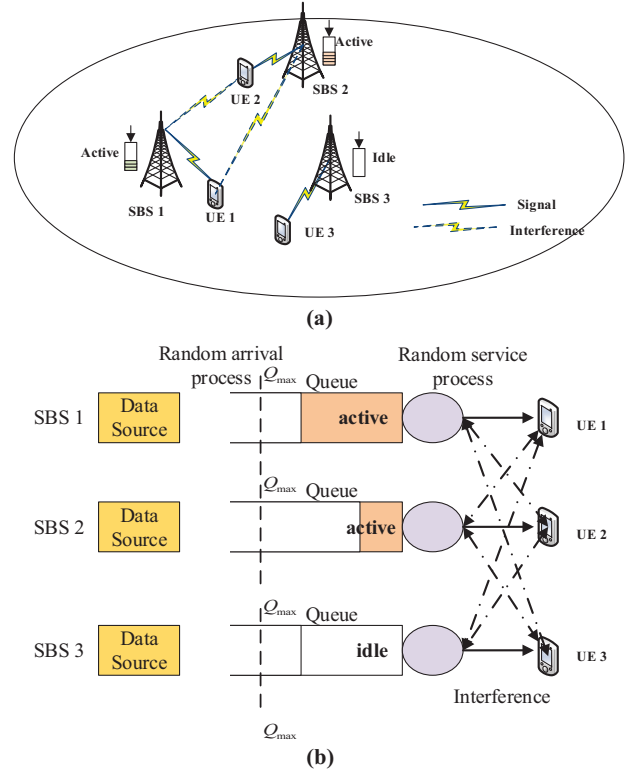}
\vspace{-0.4cm}
\caption{(a) An example of system model including four SBSs and four corresponding UEs, where SBS1 and SBS2 are active to transmit the data, and SBS3 is idle because of the idle queue; (b) An example of queueing model.}
\label{fig1}
\vspace{-0.2cm}
\end{figure}

We consider the downlink of a UDN, composed of SBSs and their associated UEs, as shown in Fig. \ref{fig1}.  Let $\mathcal{N} = {\{1,...,N\}}$ denote the set of SBSs. With an focus on the EC of the UDN, we assume that the association between SBSs and UEs are in place. We also assume that all SBSs operate in the same channel, serving one UE per SBS.

A First-Input-First-Output (FIFO) queue is used to buffer data traffic at each SBS for its associated UE. The traffic data requested by the UE randomly arrives at the queue. The queue is assumed to be sufficiently long, so that neither does it overflow nor incurs packet losses. The queues of some SBSs can be empty and the SBSs become idle over periods with no traffic arrival for their UEs. When there are no data to transmit, a SBS can be turned off to prevent interference. Let $P_{n}$ denote the idle probability of the $n$-th SBS, and $\mathcal{M}_n$ denote the set of interfering SBSs for the $n$-th SBS. $(\mathcal{N}-\mathcal{M}_n)$ is the set of idle SBSs for the $n$-th SBS.

The wireless channel from a UE to its associated SBS is modeled by Rayleigh block fading. In other words, the channels remain unchanged during a time frame $T_s$, and changes between frames. The SINR from the $n$-th SBS to its associated UE can be given by
\begin{equation}\label{SINR}
\begin{aligned}
{\gamma _n} &= \frac{{{S_{n,n}}}}{{{I_{\mathcal{M}_n}} + {\sigma ^2}}} = \frac{{{S_{n,n}}}}{{\sum\limits_{j \in \mathcal{M}_n} {{S_{j,n}}}  + {\sigma ^2}}}\\
&= \frac{{{\alpha _{n,n}}{{\left| H_{n,n} \right|}^2}}}{{\sum\limits_{j \in \mathcal{M}_n} {{\alpha _{j,n}}{{\left| H_{j,n} \right|}^2}}  + {\sigma ^2}}},
\end{aligned}
\end{equation}
where $S_{j,n}$ is the received signal strength from the $j$-th SBS to the UE associated with the $n$-th SBS, $I_{\mathcal{M}_n}$ is the interference to the UE, ${{\sigma ^2}}$ is the power of the additive white Gaussian noise, $\alpha_{j,n}$ captures the transmission power, large-scale path loss, and shadowing effect from the $j$-th SBS to the UE associated with the $n$-th SBS, ${{H_{j,n}}}$ is the Rayleigh fading coefficient with unit variance, and ${{{\left| {{H_{j,n}}} \right|}^2}}$ has an independent and identically distributed (i.i.d) exponential distribution with unit mean, i.e., ${\left| {{H_{j,n}}} \right|^2} \sim\exp (1)$. Let $\mathbf{H} = \{ {H_{j,n}},j \in \mathcal{N},n \in \mathcal{N}\} $.
For notational simplicity, we define
\begin{equation}
{\beta_{j,n}} = \frac{{{S_{j,n}}}}{{{\sigma ^2}}} = \frac{{{\alpha_{j,n}}{{\left| {{H_{j,n}}} \right|}^2}}}{{{\sigma ^2}}} = {{\bar \beta }_{j,n}}{\left| {{H_{j,n}}} \right|^2},
\end{equation}
where ${{\bar \beta }_{j,n}} = \frac{{{\alpha _{j,n}}}}{{{\sigma ^2}}}$. As a result, (\ref{SINR}) can be rewritten as
\begin{equation}
{\gamma _n} = \frac{{{\beta_{n,n}}}}{{\sum\limits_{j \in \mathcal{M}_n} {{\beta _{j,n}}}  + 1}}.
\end{equation}

For mathematic tractability, we assume that the amout of traffic (in bits) arriving at slot $t=1,2,3,\ldots$, denoted by $a_n[t]$, follows an i.i.d Bernoulli process with the arrival probability $p_n$, and $a_n[t]$ follows an exponential distribution with mean $\bar L_n$. The average arrival rate is $\mu_n  = \frac{{{{\bar L}_n}{p_n}}}{{{T_s}}}$.

\section{Effective Capacity of unsaturated UDN}
\subsection{Analysis of Effective Capacity}
In the context of EC, QoS is characterized statistically by employing the QoS exponent $\bm \theta  = \left\{ {{\theta _n},n = 1, \cdots ,N} \right\}$, as given by \cite{EC}
\begin{equation}\label{1a1}
\theta_n= - \mathop {\lim }\limits_{{Q_n^{\rm th}} \to \infty } \frac{{\log (\Pr \{ Q_n[\infty] > {Q_n^{\rm th}}\} )}}{{{Q_n^{\rm th}}}},
\end{equation}
where $Q_n[t]$ is the length of the FIFO queue at the $n$-th SBS at slot $t$, $Q_n^{\rm th}$ is the threshold of the queue length to guarantee the QoS of the traffic, and $\Pr \{ Q_n[\infty ] > {Q_n^{\rm th}}\}$ is the QoS-violation probability that the queue length exceed $Q_n^{\rm th}$. In this sense, $\theta_n$ provides the exponential decaying rate of the probability that the threshold is violated.

For the arrival process $a_n[t]$ at the $n$-th SBS, the effective bandwidth, denoted by $A_n(\theta_n)$, specifies the minimum constant service rate with guaranteed QoS requirement. Given the Bernoulli traffic arrival, the effective bandwidth of the traffic can be calculated as \cite{ChenEC}
\begin{equation}\label{EB}
\begin{aligned}
{A_n}({\theta _n}) &= \mathop {\lim }\limits_{t \to \infty } \frac{1}{{{\theta _n}t}}\log (\mathbb{E}\{ {e^{{\theta _n}\sum\limits_{i = 0}^t {{a_n}[i]} }}\} )\\
&= \frac{1}{{{\theta _n}{T_s}}}\log (\frac{{{p_n}}}{{1 - {\theta _n}{{\bar L}_n}}} + 1 - {p_n}).
\end{aligned}
\end{equation}
Given a service process $r_n[i]$ of the $n$-th SBS, the EC, denoted by $C_n(\theta_n)$, specifies the maximum, consistent, steady-state arrival rate at the input of the FIFO queue, as given by \cite{EC}
\begin{equation}\label{effective capacity}
{C_n}({\theta _n}) =  - \mathop {\lim }\limits_{t \to \infty } \frac{1}{{{\theta _n}t}}\log (\mathbb{E}\{ {e^{ - {\theta _n}\sum\limits_{i = 0}^t {r_n[i]} }}\} ).
\end{equation}
where $r_n[i]$ is the instantaneous maximum transmit rate at slot $i$ and $\mathbb{E}\{\cdot\}$ takes expectation. Given the i.i.d Rayleigh block fading, the EC can be rewritten as
\begin{equation}
{C_n}({\theta _n}) =  - \frac{1}{{{\theta _n T_s}}}\log (\mathbb{E}\{ {e^{ - {\theta _n}r_n T_s}}\} ),
\end{equation}
where ${r_n} = B\log (1 + {\gamma _n})$ is the instantaneous transmission rate of the physical layer and $B$ is the bandwidth of the channel. Note that ${r_n}$ is a random variable, depending on small-scale fading and the idle/active probability of other SBSs.

We note that if the assumptions of the Gartner-Ellis theorem hold \footnote{The Gartner-Ellis theorem assumptions require the function (i) exists for all real $\theta$, (ii) is strictly convex, and (iii) is essentially smooth.} \cite{ChenEC,7061966}  and there is a unique QoS exponent $\theta^*_n$ that satisfies
\begin{equation} \label{Q1}
{A_n}({\theta^*_n}) = {C_n}({\theta^*_n}),
\end{equation}
then the queue of the $n$-th SBS is stable and $P({Q_n}[\infty ] \ge Q_n^{{\rm{th}}})$ can be approximated to \cite{ChenEC}
\begin{equation} \label{Q}
P({Q_n}[\infty ] \ge Q_n^{{\rm{th}}}) \approx {\eta_n}\exp ( - {\theta^*_n}Q_n^{{\rm{th}}}),
\end{equation}
where $\eta_n$ is the probability of nonempty buffer at the $n$-th SBS. Given a delay bound $D_{\max}^n$, the probability that the steady-state traffic delay at the UE of the $n$-th SBS exceeds $D_{\max}^n$ can be given by \cite{7061966}
\begin{equation}
{P^n_{th}} \approx  {\eta _n}{e^{ - {\theta^*_n}{A_n}({\theta^*_n }){D^n_{\max }}}}.
\end{equation}
Given the Bernoulli traffic arrival per SBS, $\eta_n$ can be approximated to ${\eta _n} \approx 1 - \theta _n^*\bar L_n$ \cite{ChenEC}. Let ${d_n} = {\theta_n^*}{{\bar L}_n}$ for notation simplicity. The delay-violation probability can be rewritten as
\begin{equation} \label{Pth}
P_{th}^n = (1 - {d_n}){e^{ - \frac{1}{{{T_s}}}\log (\frac{{{p_n}}}{{1 - {d_n}}} + 1 - {p_n})D_{\max }^n}}.
\end{equation}

The idle probability of the $n$-th SBS is equivalent to the probability that no traffic arrives at the SBS meanwhile the buffer is empty. The probabilities of traffic arrival and empty buffer can be assumed to be independent under unsaturated or light traffic conditions. Then, $P_n$ can be expressed as
\begin{equation} \label{idle}
{P_n} \thickapprox (1 - {\eta _n})(1 - {p_n}) = \theta _n^*{{\bar L}_n}(1 - {p_n}).
\end{equation}

Note that ${P_n},\forall n \in \mathcal{N}$ is a steady-state probability. According to (\ref{Q1}) and (\ref{idle}), $P_n$ depends on the traffic arrival process at the $n$-th SBS, the idle probability of the other SBSs, ${P_j},\forall j \in N,j \ne n$, and the channel condition of all SBSs.

To analyze the EC of the UDN, we first investigate a simple scenario which involves only two SBSs and two corresponding UEs. Then this scenario is extended to a general scenario involving more SBSs and UEs, as to be discussed later.
\subsection{EC for Two-SBS Case}
Consider two UEs and two corresponding SBSs, namely SBS 1 and SBS 2. For SBS 1, there are two possible scenarios that the SINR of the associated UE undergoes.

\emph{Type I}: SBS 2 is idle. Then, the corresponding SINR is given by ${\gamma_1} = {\beta_{1,1}}$. Due to ${\beta _{1,1}} \sim \exp ({\bar \beta _{1,1}})$, the PDF of the SINR can be written as
      \begin{equation} \label{TypeI}
      f_1(x) = \frac{1}{{{{\bar \beta }_{1,1}}}}{e^{ - \frac{x}{{{{\bar \beta }_{1,1}}}}}}.
      \end{equation}

\emph{Type II}: SBS 2 is active. Then, the corresponding SINR is given by ${\gamma _1} = \frac{{{\beta _{1,1}}}}{{{\beta _{2,1}} + 1}}$. The PDF of the SINR can be written as:
      \begin{equation}\label{TypeII}
      f_2(x) = \left( {\frac{1}{{{{\bar \beta }_{1,1}} + {{\bar \beta }_{2,1}}x}} + \frac{{{{\bar \beta }_{1,1}}{{\bar \beta }_{2,1}}}}{{{{\left( {{{\bar \beta }_{1,1}} + {{\bar \beta }_{2,1}}x} \right)}^{\rm{2}}}}}} \right){e^{ - \frac{x}{{{{\bar \beta }_{1,1}}}}}}.
      \end{equation}
      \begin{proof}
      See Appendix A.
      \end{proof}
Thus, the EC of SBS 1 can be given by
\begin{equation} \label{EC1}
\begin{aligned}
{C_1}({\theta _1},P_2) =&  - \frac{1}{{{\theta _1 T_s}}}\log (\mathbb{E}\{ {e^{ - {\theta _1}{r_1 T_s}}}\} )\\
 =& - \frac{1}{{{\theta _1}{T_s}}}\log ({P_2}\int_0^{ + \infty } {{e^{ - {\theta _1}B{T_s}\log (1 + x)}}{f_1}(x)dx} \\
 &+ \left( {1 - {P_2}} \right)\int_0^{ + \infty } {{e^{ - {\theta _1}B{T_s}\log (1 + x)}}{f_2}(x)dx} ).
\end{aligned}
\end{equation}
Similarly, the EC of SBS 2 can be given by
\begin{equation} \label{EC2}
\begin{aligned}
{C_2}({\theta _2},P_1) =&  - \frac{1}{{{\theta _2 T_s}}}\log (\mathbb{E}\{ {e^{ - {\theta _2}{r_2 T_s}}}\} )\\
 =& - \frac{1}{{{\theta _2}{T_s}}}\log ({P_1}\int_0^{ + \infty } {{e^{ - {\theta _2}B{T_s}\log (1 + x)}}{f_3}(x)dx}  \\
 & + \left( {1 - {P_1}} \right)\int_0^{ + \infty } {{e^{ - {\theta _2}B{T_s}\log (1 + x)}}{f_4}(x)dx} ),
\end{aligned}
\end{equation}
where $f_3(x)$ and $f_4(x)$ are the PDFs of the SINR at SBS 2 in regards of the two types, and can be obtained by switching the roles of SBSs 1 and 2 in (\ref{TypeI}) and (\ref{TypeII}), respectively.

Substituting (\ref{idle}) into (\ref{EC1}) and (\ref{EC2}), we can obtain the EC of each of the SBSs, as given by,
\begin{equation} \label{1}
{{C_1}({\theta _1},\theta _2^*) = {C_1}\left( {{\theta _1},\theta _2^*{{\bar L}_2}(1 - {p_2})} \right)},
\end{equation}
\begin{equation} \label{2}
{C_2}(\theta _1^*,{\theta _2}) = {C_2}\left( {{\theta _2},\theta _1^*{{\bar L}_1}(1 - {p_1})} \right).
\end{equation}
\subsection{EC for N--SBS Case}
We proceed to consider a UDN with $N$ SBSs. For any SBS, a total of $2^{(N-1)}$ sets are needed to account for all possible cases in terms of the ``on'' and ``off'' states of all other SBSs. The EC of the $n$-th SBS can be given by
\begin{equation} \label{q}
\begin{aligned}
{C_n}&({\theta _n},{P_1}, \ldots {P_{n - 1}},{P_{n + 1}}, \ldots ,{P_N})\\
&=  - \frac{1}{{{\theta _n}{T_s}}}\log ({\mathbb{E}_{{\cal M}_n, \mathbf{H}}}\{ {e^{ - {\theta _n}{r_n}{T_s}}}\} ),
\end{aligned}
\end{equation}
where ${\mathbb{E}_{{\cal M}_n, \mathbf{H}}}\{  \cdot \} $ takes expectation over all other active sets and channel fading conditions. Despite the discussion of the two--SBS case can be readily applied to this case, unfortunately the PDF of SINR is intractable, due to multiple integrals accounting for $(N-1)$ interfering SBSs per SBS.

We propose to decouple ${\mathbb{E}_{{\cal M}_n, \mathbf{H}}}\{  \cdot \} $, as given by
\begin{equation} \label{ECtem}
\begin{array}{l}
{\mathbb{E}_{{\cal M}_n,{\bf{H}}}}\{ {e^{ - {\theta _n}B{T_s}\log (1 + {\gamma _n})}}\} \\
\mathop {\rm{ = }} {\mathbb{E}_{\mathcal{M}_n}}\left\{ {{\mathbb{E}_{\bf{H}}}\{ {e^{ - {\theta _n}B{T_s}\log (1 + {\gamma _n})}}|\mathcal{M}_n\} } \right\}\\
\mathop {\rm{ = }}\limits^{(a)} {\mathbb{E}_{\mathcal{M}_n}}\left\{ {{\mathbb{E}_{\bf{H}}}\{ {e^{ - {\theta _n}B{T_s}\log (1 + {\gamma _n})}}\} } \right\}\\
\mathop  = \limits^{(b)} {E_{{\cal M}_n}}\left\{ {\int_0^{ + \infty } {\Pr ({e^{ - {\theta _n}B{T_s}\log (1 + {\gamma _n})}} \ge t)} dt} \right\}
\end{array}
\end{equation}
where (a) is due to the independence of other active sets and channel fading conditions, (b) is due to the fact that $\mathbb{E}\left\{ X \right\} = \int_0^{ + \infty } {P(X \ge t)} dt$, and the probability ${\Pr ({e^{ - {\theta _n}B{T_s}\log (1 + {\gamma _n})}} \ge t)}$ is given by
\begin{equation} \label{nanting}
\begin{array}{l}
\Pr ({e^{ - {\theta _n}B{T_s}\log (1 + {\gamma _n})}} \ge t)\\
 = \Pr ({\gamma _n} \le {2^{ - \frac{1}{{{\theta _n}B{T_s}}}\ln t}} - 1)\\
 = \Pr (\frac{{{\beta _{n,n}}}}{{\sum\limits_{j \in {\cal M}_n} {{\beta _{j,n}}}  + 1}} \le {2^{ - \frac{1}{{{\theta _n}B{T_s}}}\ln t}} - 1)\\
 = 1 - \Pr \left( {{{\left| {{H_{n,n}}} \right|}^2} \ge \frac{{({2^{ - \frac{1}{{{\theta _n}B{T_s}}}\ln t}} - 1)(\sum\limits_{j \in {\cal M}_n} {{\beta _{j,n}}}  + 1)}}{{{{\bar \beta }_{n,n}}}}}\right)\\
\mathop  = \limits^{(a)} 1 - {\mathbb{E}_{{\bf{H}}}}\left\{ {\exp \left( { - s(\sum\limits_{j \in \mathcal{M},j \ne n} {{\beta _{j,n}}}  + 1)} \right)} \right\}\\
\mathop {\rm{ = }}\limits^{(b)} 1 - {{\rm{e}}^{ - s}}{\prod\limits_{j \in \mathcal{M}_n} {{\mathbb{E}_{\bf{H}}}\left\{ {\exp \left( { - s{\beta _{j,n}}} \right)} \right\}} } \\
{\rm{ = }}1 - {{\rm{e}}^{ - s}}\prod\limits_{j \in {{\cal M}_n}} {\int_{\rm{0}}^{{\rm{ + }}\infty } {{e^{ - s{{\bar \beta }_{j,n}}x}}} {e^{ - x}}dx} \\
\mathop {= } 1 - {{\rm{e}}^{ - s}}{\prod\limits_{j \in \mathcal{M}_n} {\frac{1}{{1 + s{{\bar \beta }_{j,n}}}}} }.
\end{array}
\end{equation}
In (\ref{nanting}), (a) is due to the fact that the channel model  ${\left| {{H_{n,n}}} \right|^2} \sim\exp (1)$. For notational simplicity, we denote $s = \frac{{{2^{ - \frac{1}{{{\theta _n}B{T_s}}}\ln t}} - 1}}{{{{\bar \beta }_{n,n}}}}$. (b) is due to the independence of the fading channel coefficients.

Given the fact that ${\gamma _n} > 0$, in the case of $t>1$, we have
\begin{equation} \label{chunye}
\begin{array}{l}
\Pr ({e^{ - {\theta _n}B{T_s}\log (1 + {\gamma _n})}} \ge t)\\
 = \Pr ({\gamma _n} \le {2^{ - \frac{1}{{{\theta _n}B{T_s}}}\ln t}} - 1)\\
 = 0.
\end{array}
\end{equation}
Substituting (\ref{nanting}) and (\ref{chunye}) into (\ref{ECtem}), we have
\begin{equation}
\begin{array}{l}
{\mathbb{E}_{{\mathcal{M}_n},{\bf{H}}}}\{ {e^{ - {\theta _n}B{T_s}\log (1 + {\gamma _n})}}\}  \\
= {\mathbb{E}_{{\mathcal{M}_n}}}\left\{ {\int_0^{ + \infty } {1 - {{\rm{e}}^{ - s}}\prod\limits_{j \in {{\cal M}_n}} {\frac{1}{{1 + s{{\bar \beta }_{j,n}}}}} } dt} \right\}\\
 = 1 - {\mathbb{E}_{{\mathcal{M}_n}}}\left\{ {\int_0^1 {{{\rm{e}}^{ - s}}\prod\limits_{j \in {{\cal M}_n}} {\frac{1}{{1 + s{{\bar \beta }_{j,n}}}}} } dt} \right\}
\end{array}
\end{equation}
According to the Fubini's theorem \cite{Fubinitheorem}, because ${{{\rm{e}}^{ - s}}\prod\limits_{j \in {M_n}} {\frac{1}{{1 + s{{\bar \beta }_{j,n}}}}} }$ is nonnegative and ${\mathbb{E}_{{{\cal M}_n}}}\left\{ {\int_0^1 {{{\rm{e}}^{ - s}}\prod\limits_{j \in {M_n}} {\frac{1}{{1 + s{{\bar \beta }_{j,n}}}}} } dt} \right\} < \infty $, we have
\[\begin{array}{l}
{\mathbb{E}_{{{\cal M}_n}}}\left\{ {\int_0^1 {{{\rm{e}}^{ - s}}\prod\limits_{j \in {\mathcal{M}_n}} {\frac{1}{{1 + s{{\bar \beta }_{j,n}}}}} } dt} \right\}\\
 = \int_0^1 {{\mathbb{E}_{{{\cal M}_n}}}\left\{ {{{\rm{e}}^{ - s}}\prod\limits_{j \in {\mathcal{M}_n}} {\frac{1}{{1 + s{{\bar \beta }_{j,n}}}}} } \right\}} dt.
\end{array}\]
As a result, ${\mathbb{E}_{{\mathcal{M}_n},{\bf{H}}}}\{ {e^{ - {\theta _n}B{T_s}\log (1 + {\gamma _n})}}\}$ can be rewritten as
\begin{equation} \label{1123}
\begin{array}{l}
{\mathbb{E}_{{{\cal M}_n},{\bf{H}}}}\{ {e^{ - {\theta _n}B{T_s}\log (1 + {\gamma _n})}}\} \\
 = 1 - \int_0^1 {{{\rm{e}}^{ - s}}{\mathbb{E}_{{{\cal M}_n}}}\left\{ {\prod\limits_{j \in {M_n}} {\frac{1}{{1 + s{{\bar \beta }_{j,n}}}}} } \right\}} dt
\end{array}
\end{equation}
We introduce new auxiliary random variables: for $j=1, \ldots, N$,
\[{X_j} = \left\{ {\begin{array}{*{20}{c}}
{\frac{1}{{1 + s{{\bar \beta }_{j,n}}}},{\rm{when~SBS~}}j{~\rm{~is~active}}}\\
{1{\rm{~~~~~~,when~SBS~}}j{~\rm{~is~idle}}}
\end{array}} \right.\]
and let $Y_n = \prod\limits_{j \in \mathcal{N},j \ne n} {{X_j}}$. Due to the independence of $X_j, \forall j \in \mathcal{N}$, we have
\begin{equation} \label{transform}
{\mathbb{E}_{\cal M}}\left\{ {\prod\limits_{j \in \mathcal{M}_n} {\frac{1}{{1 + s{{\bar \beta }_{j,n}}}}} } \right\} = \mathbb{E}_{Y_n}\left\{ {{Y_n}} \right\} = \prod\limits_{j \in \mathcal{N},j \ne n} {\mathbb{E}_{X}\left\{ {{X_j}} \right\}}.
\end{equation}
Note that (\ref{transform}) transforms the expectation over all other active sets into the expectation over each SBS, which reduce the computational complexity.

Substituting (\ref{transform}) into (\ref{1123}), (\ref{1123}) can be rewritten as
\begin{equation} \label{ECtem1}
\begin{array}{l}
{{\mathbb{E}_{{{\cal M}_n},{\bf{H}}}}\{ {e^{ - {\theta _n}B{T_s}\log (1 + {\gamma _n})}}\} }\\
{ = 1 - \int_0^1 {{{\rm{e}}^{ - s}}\prod\limits_{j \in \mathcal{N},j \ne n} {\left( {(1 - {P_j})\frac{1}{{1 + s{{\bar \beta }_{j,n}}}} + {P_j}} \right)} dt} }
\end{array}
\end{equation}
Substituting (\ref{ECtem1}) into (\ref{q}), we can obtain the EC of each of the SBSs, as given by,
\begin{equation} \label{NEC}
\begin{array}{*{20}{l}}
{{C_n}({\theta _n},{P_1}, \ldots {P_{n - 1}},{P_{n + 1}}, \ldots ,{P_N})}\\
{ =  - \frac{1}{{{\theta _n}{T_s}}}\log \left\{ {1 - \int_0^1 {{{\rm{e}}^{ - s}}\!\!\!\!\!\!\prod\limits_{j \in \mathcal{N},j \ne n} {\left( {(1 - {P_j})\frac{1}{{1 + s{{\bar \beta }_{j,n}}}} + {P_j}} \right)} dt} } \right\}}
\end{array}
\end{equation}
Substituting (\ref{idle}) into (\ref{NEC}), we can obtain the ${C_n}({\theta _n},\theta _1^*, \ldots \theta _{n - 1}^*,\theta _{n + 1}^*, \ldots ,\theta _N^*)$. The following theorem and corollaries could reveal the key properties of EC, which can be used to configure and optimize network operations, such as resource allocation and admission control with non-trivial QoS requirements.
\begin{thm}
As the QoS exponent $\theta_n$ increases, $C_n$ decreases and $\bm C_{-n}$ increases, where $\bm C_{-n} = \{C_{j},j = 1, \cdots ,N, j\ne n \}$ collects the EC of the rest of the SBSs except the $n$-th SBS.
\end{thm}
\begin{proof}
$C_n$ can be proven to be a monotonically decreasing function of $\theta_n$ \cite{CuiEC}. According to (\ref{idle}), we can observe that the idle probability $P_n$ monotonically increases with $\theta_n$. Then, the interference of the $n$-th SBSs to the other SBSs decreases. As a result, $\bm C_{-n}$ increases with $\theta_n$.
\end{proof}

\begin{cor}
Given $\bar L_n, n \in \{ 1,2,\cdots ,N\}$, the EC of the SBSs satisfy the following nonlinear equations with unknown $\bm \theta$,
\begin{equation} \label{eq1}
\left\{ {\begin{array}{*{20}{l}}
~{{C_1}\left( \bm \theta  \right) = {A_1}({\theta _1});}\\
\begin{array}{l}
{C_2}\left( \bm \theta  \right) = {A_2}({\theta _2});\\
 \vdots \\
{C_N}\left( \bm \theta  \right) = {A_N}({\theta _N}).
\end{array}
\end{array}} \right.
\end{equation}
\end{cor}
Eq. (\ref{eq1}) can be solved by standard numerical methods, such as the Newton's method. Nevertheless, according to Theorem 1 and that effective bandwidth monotonically increases with $\theta_n$, we can develop a fast algorithm to bisectionally search for the solution, as summarized in Algorithm 1. The proposed algorithm is efficient since it requires only linear search, instead of calculating the Jacobian matrix of (\ref{eq1}) (as typically required in the Newton's method).

\begin{cor}
Given the QoS requirements $d_n ={\theta_n^*}{{\bar L}_n}, n \in \{ 1,2,\cdots ,N\}$ (According to (\ref{Pth}), $d_n$ can decide on the delay requirements), the maximum allowed average arrival rate is ${\mu _n} = \frac{{{d_n}{p_n}}}{{\theta _n^*{T_s}}}$, where the QoS exponent $\theta _n^*$ can be obtained by solving the following equations.
\begin{equation} \label{sita1}
{\theta _n}{T_s}{C_n}\left( {{\theta _n},d_1^*, \ldots d_{n - 1}^*,d_{n + 1}^*, \ldots ,d_N^*} \right) \!\!= \log (\frac{{{p_n}}}{{1 - {d_n}}} + 1 - {p_n}).
\end{equation}
\end{cor}
Note that (\ref{sita1}) is obtained by substituting (\ref{EB}) and (\ref{NEC}) into (\ref{Q1}), and can be solved using a bisection method. This is because, according to (\ref{q}), we have
\[\begin{array}{l}
{\theta _n}{T_s}{C_n}\left( {{\theta _n},d_1^*, \ldots d_{n - 1}^*,d_{n + 1}^*, \ldots ,d_N^*} \right)\\
 =  - \log ({E_{{M_n},{\bf{H}}}}\{ {e^{ - {\theta _n}{r_n}{T_s}}}\} ).
\end{array}\]
Obviously, ${\theta _n}{T_s}{C_n}\left( {{\theta _n},d_1^*, \ldots d_{n - 1}^*,d_{n + 1}^*, \ldots ,d_N^*} \right)$ monotonically increase with respect to $\theta_n$.

\begin{cor}
As ${\theta _n} \to 0, \forall n \in \mathcal{N}$, ${P_n} \to 0, \forall n \in \mathcal{N}$ according to (\ref{idle}). In this case, the traffic is saturated and the EC is expected to approach the classical capacity with saturated traffic.
\end{cor}
\subsection{Maximization of Total EC Impacted by Traffic Unsaturation}
The effective capacity of a interference-limited network has to be further developed, due to sophisticated interactions among SBSs operating in the same frequency. More specifically, the effective capacity of a designated SBS decreases with the increasing traffic of other co-channel SBSs. This is because the increasing traffic of other SBSs can result in severe interference towards the designated SBS.

We proceed to study the impact of the level of traffic saturation, i.e., $p_n$, on the total effective capacity of UDN. A non-cooperative game can be formulated to characterize the sophisticated interactions in the UDN. The game is designed to maximize the amount of traffic while guaranteeing the QoS of the traffic, as given by ${\cal G} = \left[ {{\cal N},{{\left\{ {{{p}_n}} \right\}}_{n \in {\cal N}}},{{\left\{ {{U_n}} \right\}}_{n \in {\cal N}}}} \right]$, where ${\cal N} = \{ 1, \ldots ,N\} $ is a set of players (i.e., SBSs), ${{p}_n}$ is traffic unsaturation policy, and ${{U_n}}$ is a utility (payoff) function of player $n$. Let ${\textbf{p}_{ - n}} = \{ {p_1}, \ldots {p_{n - 1}},{p_{n + 1}}, \ldots ,{p_N}\}$. ${U_n}({p_n},{\textbf{p}_{ - n}})$ is given by
\[
\begin{aligned}
&{U_n}({p_n},{\textbf{p}_{ - n}}) = \frac{{{L_n}{p_n}}}{{{\theta _n}{T_s}}}\\
&s.t.{A_n}({p_n}) \le {C_n}(\textbf{p}_{ - n}),\\
&~~~~0 \le {p_n} \le 1\\
\end{aligned}\]
where the objective represents the maximum allowed average arrival rate, and the first constraint is set to guarantee the QoS.

\begin{cor} The best-response function of player $n$ can be given by
\[b({\textbf{p}_{-n}}) = \left[ {\frac{{\left( {{e^{{\theta _n}{T_s}{C_n}(\bf p_{ - n})}} - 1} \right)(1 - {\theta _n}{L_n})}}{{{\theta _n}{L_n}}}} \right]_0^1.\]
where $\left[ x \right]_a^b = \min [\max (x,a),b]$.
\end{cor}
\begin{proof}
According to (5), (12) and (27), ${A_n}({p_n}) \le {C_n}(\textbf{p}_{ - n})$ can be rewritten as
\[{p_n} \le \frac{{\left( {{e^{{\theta _n}{T_s}{C_n}(\textbf{p}_{ - n})}} - 1} \right)(1 - {\theta _n}{L_n})}}{{{\theta _n}{L_n}}}\]
Because the best-response function of player $n$ satisfies
$b(\textbf{p}_{-n})= \left\{ {\left. {{p_n}} \right|{U_n}({p_n},{\textbf{p}_{ - n}}) \ge {U_n}({{p'}_n},{\textbf{p}_{ - n}})} \right\}$, we have
\[b({\bf{p}_{ - n}}) = \max \left( {\frac{{\left( {{e^{{\theta _n}{T_s}{C_n}({\bf{p}_{ - n}})}} - 1} \right)(1 - {\theta _n}{L_n})}}{{{\theta _n}{L_n}}},1} \right).\]
The proof is accomplished.
\end{proof}
A pure-strategy NE of the non-cooperative game ${\cal G}$ is a strategy profile ${{\left\{ {{{p}_n^*}} \right\}}_{n \in {\cal N}}}$ such that $\forall n \in \cal N$ we have the following:
\[{U_n}(p_n^*,\textbf{p}_{ - n}^*) \ge {U_n}({p_n},\textbf{p}_{ - n}^*).\]
Using Corollary 4, we can obtain
\[p_n^* = b(\textbf{p}_{ - n}^*),\forall n \in \mathcal{N}.\]
According to \cite{R1}, we can iteratively search for the NE of the traffic unsaturation policy.

\begin{algorithm}[t]
  \caption{A heuristic algorithm to solve (\ref{eq1}).}
  \label{Online Algorithm}
  \begin{algorithmic}[1]
  \STATE Given the lower bound $LB = [0,0,\ldots,0]$, the upper bound $UB =[\frac{1}{{{{\bar L}_1}}},\frac{1}{{{{\bar L}_2}}},\ldots,\frac{1}{{{{\bar L}_N}}}]$, and $\bm \theta = \frac{{LB + UB}}{2}$.
  \REPEAT
  \FOR{$n \in \mathcal{N}$}
   \IF{$C_n(\bm \theta) > A_n(\theta_n)$}
  \STATE  $LB[n] = {\theta _n}$, and ${\theta _n} = \frac{{LB[n] + UB[n]}}{2}$,
  \ELSE
  \STATE  $UB[n] = {\theta _n}$, and ${\theta _n} = \frac{{LB[n] + UB[n]}}{2}$.
  \ENDIF
  \ENDFOR
  \UNTIL{$\left| {{C_n}(\theta ) - {A_n}(\theta )} \right| \ll \xi ,\forall n \in \mathcal{N}$, where $0 < \xi  \ll 1$.}
  \end{algorithmic}
\end{algorithm}

\section{Simulation Result}
In this section, simulations are carried out to verify the accuracy of our proposed model and evaluate the performance of UDNs, where $N$ SBSs with overlapping coverage are uniformly distributed in an area of ${\rm{500}} \times {\rm{500}}~{{\rm{m}}^2}$, and 1000 UEs are randomly and uniformly distributed within the area. In the simulations there is one channel and each SBS serves a single UE in a channel of every instant. The locations of the SBSs follow a spatial Poisson process. The coverage of each of the SBSs can then be specified to be the area, in which the received signals of the SBS are stronger than those from the rest of the SBSs. The user served by the SBS in a channel is set to be uniformly randomly distributed within the area. And the traffic arrival probability of the UE is $p_n= 0.2$, the mean of traffic arrival per slot is $\bar L_n=100$ bits, the queue length threshold is $Q_n^{\rm th} = 10$, the channel bandwidth is  $B=1$ MHz, and the duration of a slot is $T_s = 1$ ms. The transmit power of each SBS is 23 dBm. The noise spectral density is $-174$ dBm/Hz. The Large-scale path loss is $PL_n = 60+37.6\mathrm{lg}(D_n)$, where $D_n$ is the distance from the $n$-th SBS to its associated UE. Each point of figures is the average of 2000 independent runs and for every independent run, 100000 time slots are tested in the simulation.

In Fig. \ref{fig2}, the accuracy of the proposed model for 2--SBS case is verified by comparing with simulation results and the full interference model, where there are two UEs and two corresponding SBSs, ${\bm{\bar \beta }}{\rm{ = [10,1;2,20]}}$ dB. The full interference model means that the interference for $n$-th UE comes from all SBSs except the $n$-th SBS.
The analytical results of the queue-violation probability are obtained by (\ref{Q}), (\ref{EC1}), (\ref{EC2}) and (\ref{eq1}). As traffic increases, the  the queue-violation probability also increases. We can see that the analytical results coincide the simulation results for both UEs 1 and 2, in other words, our analysis is accurate. There exist big gaps between the simulation results and the full interference model, indicating that random interference resulting from stochastic traffic arrival per SBS can have a strong impact on the EC and the QoS (i.e., queue length or delay).

\begin{figure}[!t]
\centering
\includegraphics[width=3.5in]{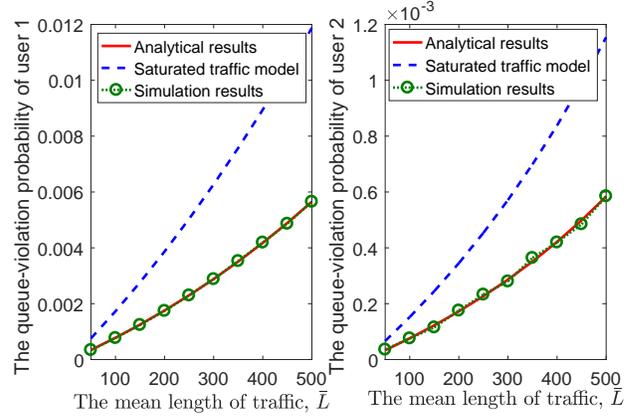}
\vspace{-0.4cm}
\caption{The QoS-violation probability of UE 1 and 2 in the two-SBS case.}
\label{fig2}
\vspace{-0.2cm}
\end{figure}

In Fig. \ref{fig3}, the accuracy of the proposed model for N--SBS case is verified by comparing with simulation results. The analytical results of the QoS-violation probability are obtained by (\ref{Q}), (\ref{NEC}) and (\ref{eq1}). We can see that the analytical results coincide the simulation results. There also exist big gaps between the simulation results and the full interference model. With the increasing density of SBSs, the queue-violation probability decreases. This is because SBSs are increasingly close to UEs, which can help improve signal qualities of the UEs. As the density of SBSs continues to increase, the queue violation probability stabilizes. This is because the interference of other SBSs also increase, counteracting the effect of the improving signal qualities.

\begin{figure}[!t]
\centering
\includegraphics[width=3.5in]{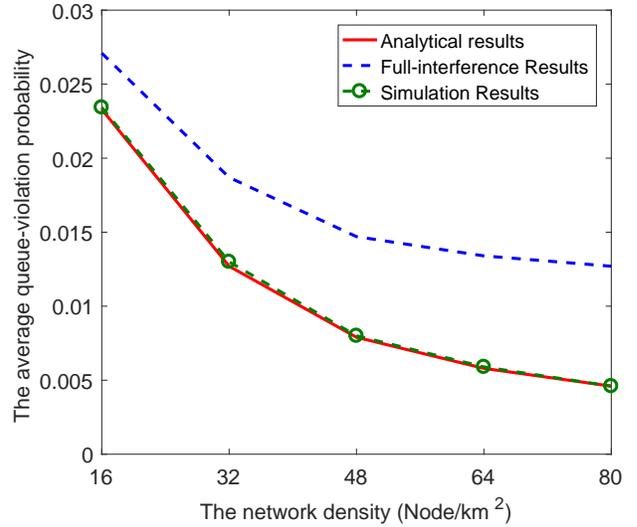}
\vspace{-0.4cm}
\caption{The average queue-violation probability of N SBSs}
\label{fig3}
\vspace{-0.2cm}
\end{figure}

\begin{figure}[!t]
\centering
\includegraphics[width=3.5in]{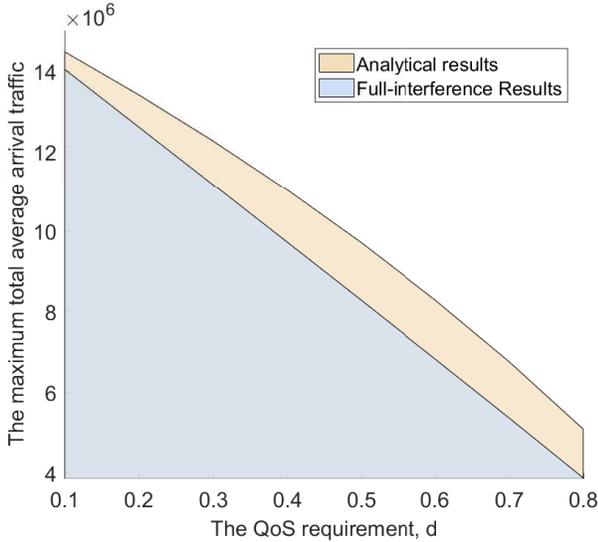}
\vspace{-0.4cm}
\caption{The maximum total average traffic arrival rate vs the QoS requirement.}
\label{fig4}
\vspace{-0.2cm}
\end{figure}

The important region of EC in UDNs is plotted in Fig. \ref{fig4}, and so is the tradeoff between the maximum total average traffic arrival rate and the QoS requirement.
Particularly, Fig. \ref{fig4} plots the maximum total average traffic arrival rate with the growing QoS requirement, where $d = {d_1} = {d_2} \cdots {\rm{ = }}{d_N}$, and $N=8$. Given the network density, the QoS and traffic amount that the SBSs can achieve in  UDNs are captured in the corresponding shaded area enclosed by a curve plotted by Corollary 2. It is shown that the total average arrival rate of our model overtops that of the full interference model, due to the consideration of random interference. As $d \to 0$, the analytical results approach to the full interference model, which is consistent with Corollary 3.

In Fig. \ref{fig5}, we evaluate the maximum traffic arrival rate of each SBS in UDNs, under different network densities and QoS requirements. We can see that the analytical results overtop the full interference model. And the gap between our model and the full interference model increases with the increasing density. This is because the dense deployment and the corresponding stochastic traffic arrival per SBS can have a stronger impact on the SINR distribution and the EC. With an increasing density of SBSs, the maximum total arrival rate of each SBS increases and gradually stabilizes. This conclusion matches that of Fig. \ref{fig3}. As the QoS requirement $d$ increases, the maximum allowed arrival rate of each SBS decreases.

\begin{figure}[!t]
\centering
\includegraphics[width=3.5in]{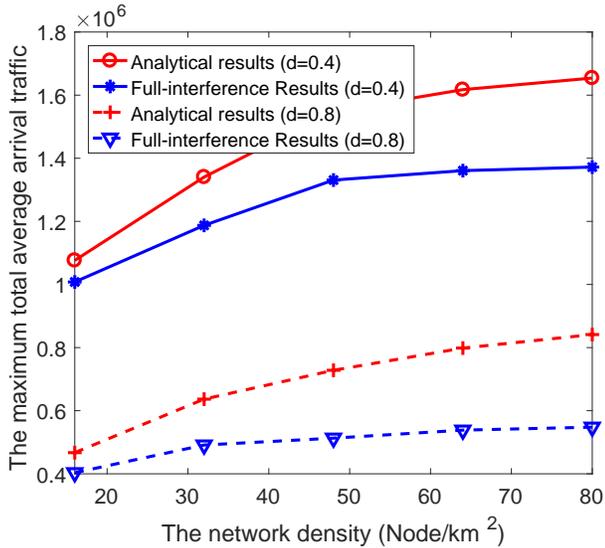}
\vspace{-0.4cm}
\caption{The maximum average traffic arrival rate of each SBS in the N-SBS case, where all UEs have the same QoS requirement, $d = {d_n},\forall n$.}
\label{fig5}
\vspace{-0.2cm}
\end{figure}

In Fig. 6, the impact of the different traffic pattern on the network performance is evaluated, where $B=0.1$ MHz and ${\theta _n} = {10^{ - 3}}$. It is shown that using Corollary 4 the game can converge to the NE point at finite iterations and the NE point is superior to other strategies. We also observe that the total traffic amount has a linear relationship with the network density. This is because the capacity of each SBS has stabilized, as mentioned earlier in Fig. 5.

\begin{figure}[!t]
\centering
\includegraphics[width=3.5in]{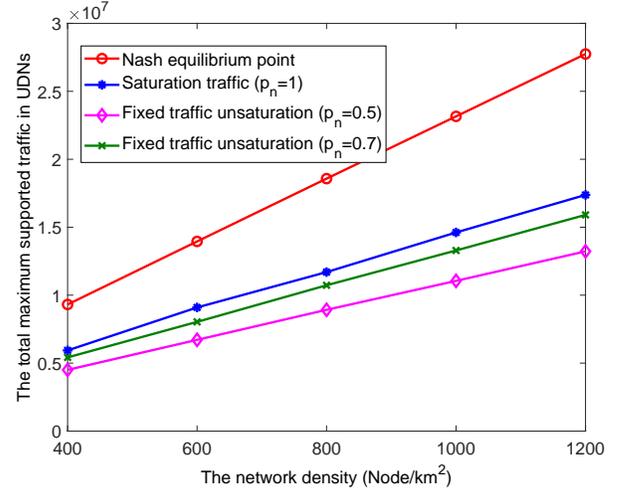}
\vspace{-0.4cm}
\caption{The different traffic pattern.}
\label{fig10}
\vspace{0.2cm}
\end{figure}

\section{Conclusion}
In this paper, we develop a cross-layer analytical model and derive the EC of UDNs under time-varing random inter-cell interference due to unsaturated traffic and statistically characterized QoS requirements. Our analysis reveals intrinsic interactions among network capacity, traffic and QoS. Validated by simulations, the proposed model is accurate. Using the proposed model, we obtain the maximum traffic arrival rate with QoS guarantee in UDNs. In our future work, we will maximize the total EC of $N$ SBSs in UDNs through optimization of bandwidth and transmit power.

\section*{Acknowledgment}
The work was supported in part by the National Nature Science Foundation of China Project under Grant 61471058, in part by the Hong Kong, Macao and Taiwan Science and Technology Cooperation Projects under Grant 2016YFE0122900, in part by the Beijing Science and Technology Commission Foundation under Grant 201702005, in part by the 111 Project of China under Grant B16006, in part by BUPT Excellent Ph.D. Students Foundation (CX2017305) and the China Scholarship Council (CSC).

\appendices
\section{Proof of the joint PDF of the SINR in the two-SBSs case}
The joint SINR is given by ${\gamma _1} = \frac{{{\beta _{1,1}}}}{{{\beta _{2,1}} + 1}}$. Let $v = {\beta _{2,1}} + 1$. The PDF of $v$ is given by
\[{f_v}({x_2}) = \frac{1}{{{{\bar \beta }_{2,1}}}}{e^{ - \frac{{{x_2} - 1}}{{{{\bar \beta }_{2,1}}}}}}, x_2 \ge 1\]
As $\beta_{1,1}$ and $\beta_{2,1}$ are independent, the joint PDF of $(\beta _{1,1}, v)$ can be given by
\begin{equation}
{f_{{\beta _{1,1}},v}}({x_1},{x_2}) = \frac{1}{{{{\bar \beta }_{1,1}}{{\bar \beta }_{2,1}}}}{e^{ - \frac{{{x_1}}}{{{{\bar \beta }_{1,1}}}} - \frac{{{x_2} - 1}}{{{{\bar \beta }_{2,1}}}}}}.
\end{equation}
Then using the bivariate transformation theory \cite{Wang,Song}, we denote $x=\frac{{{x_1}}}{{{x_2}}}$ and $y=x_2$. The Jacobian of the transformation is
\[J = \left| {\begin{array}{*{20}{c}}
{\frac{{\partial x}}{{\partial {x_1}}}}&{\frac{{\partial x}}{{\partial {x_2}}}}\\
{\frac{{\partial y}}{{\partial {x_1}}}}&{\frac{{\partial y}}{{\partial {x_2}}}}
\end{array}} \right| = \left| {\begin{array}{*{20}{c}}
{\frac{1}{{{x_2}}}}&{ - \frac{{{x_1}}}{{{{\left( {{x_2}} \right)}^2}}}}\\
0&1
\end{array}} \right| = \frac{1}{{{x_2}}}.\]
As a result, the joint PDF of $(\gamma _1, v)$ is given by
\[\begin{aligned}
{f_{{\gamma _1},v}}(x,{x_2}) &= \frac{{{f_{{\beta _{1,1}},v}}(x{x_2},{x_2})}}{J}\\
& = \frac{{{x_2}}}{{{{\bar \beta }_{1,1}}{{\bar \beta }_{2,1}}}}{e^{ - \frac{{x{x_2}}}{{{{\bar \beta }_{1,1}}}} - \frac{{{x_2} - 1}}{{{{\bar \beta }_{2,1}}}}}}.
\end{aligned}\]
Then the PDF of $\gamma _1$, $f_2(x)$, is the marginal distribution of ${f_{{\gamma _1},v}}(x,{x_2})$, as given by
\begin{equation} \label{IN2}
\begin{aligned}
{f_2}(x) &= \int_1^{ + \infty } {{f_{{\gamma _1},v}}(x,{x_2})d{x_2}} \\
& = \frac{{{e^{\frac{1}{{{{\bar \beta }_{2,1}}}}}}}}{{{{\bar \beta }_{1,1}}{{\bar \beta }_{2,1}}}}\int_1^{ + \infty } {{x_2}{e^{ - \left( {\frac{x}{{{{\bar \beta }_{1,1}}}}{\rm{ + }}\frac{{\rm{1}}}{{{{\bar \beta }_{2,1}}}}} \right){x_2}}}d{x_2}}.
\end{aligned}
\end{equation}
According to \cite{gradshteyn2014table}, we have
\begin{equation} \label{IN}
\int_u^{ + \infty } {{x^n}{e^{ - \mu x}}dx = {e^{ - u\mu }}\sum\limits_{k = 0}^n {\frac{{n!{u^k}}}{{k!{\mu ^{n - k + 1}}}}} }
\end{equation}
Substituting $n=1$, $u=1$ and $\mu {\rm{ = }}\frac{x}{{{{\bar \beta }_{1,1}}}}{\rm{ + }}\frac{{\rm{1}}}{{{{\bar \beta }_{2,1}}}}$ into (\ref{IN}), we can obtain
\begin{equation} \label{IN1}
\begin{array}{l}
\int_1^{ + \infty } {{x_2}{e^{ - \left( {\frac{x}{{{{\bar \beta }_{1,1}}}}{\rm{ + }}\frac{{\rm{1}}}{{{{\bar \beta }_{2,1}}}}} \right){x_2}}}d{x_2}} \\
 = {e^{ - \left( {\frac{x}{{{{\bar \beta }_{1,1}}}}{\rm{ + }}\frac{{\rm{1}}}{{{{\bar \beta }_{2,1}}}}} \right)}}\sum\limits_{k = 0}^{\rm{1}} {\frac{1}{{k!{{(\frac{x}{{{{\bar \beta }_{1,1}}}}{\rm{ + }}\frac{{\rm{1}}}{{{{\bar \beta }_{2,1}}}})}^{2 - k}}}}} \\
 = {e^{ - \left( {\frac{x}{{{{\bar \beta }_{1,1}}}}{\rm{ + }}\frac{{\rm{1}}}{{{{\bar \beta }_{2,1}}}}} \right)}}\left( {\frac{{{{\left( {{{\bar \beta }_{1,1}}{{\bar \beta }_{2,1}}} \right)}^{\rm{2}}}}}{{{{({{\bar \beta }_{1,1}} + x{{\bar \beta }_{2,1}})}^2}}}{\rm{ + }}\frac{{{{\bar \beta }_{1,1}}{{\bar \beta }_{2,1}}}}{{{{\bar \beta }_{1,1}} + x{{\bar \beta }_{2,1}}}}} \right)
\end{array}
\end{equation}

Substituting (\ref{IN1}) into (\ref{IN2}), $f_2(x)$ can be derived as
\[{f_2(x) = \left( {\frac{1}{{{{\bar \beta }_{1,1}} + {{\bar \beta }_{2,1}}x}} + \frac{{{{\bar \beta }_{1,1}}{{\bar \beta }_{2,1}}}}{{{{\left( {{{\bar \beta }_{1,1}} + {{\bar \beta }_{2,1}}x} \right)}^{\rm{2}}}}}} \right){e^{ - \frac{x}{{{{\bar \beta }_{1,1}}}}}}}.\]
This concludes the proof.

\end{document}